\documentclass{amsart}
\usepackage{amsmath, amsthm, amscd, amsfonts, amssymb, color}
 \newtheorem{thm}{Theorem}[section]
 \newtheorem{cor}[thm]{Corollary}
 
 \newtheorem{prop}[thm]{Proposition}
 \theoremstyle{definition}
 \newtheorem{defn}[thm]{Definition}
 
 \theoremstyle{remark}
 \newtheorem{rem}[thm]{Remark}
 \newtheorem{ex}[thm]{Example}
 \numberwithin{equation}{section}

\begin{document}



\title[{A description of pseudo-bosons in terms of... }]{A description of pseudo-bosons in terms of  nilpotent Lie algebras }

\author[F. Bagarello]{Fabio Bagarello}
\address{Scuola Politecnica dell'Universit\'a \endgraf
Dipartimento di Energia, Ingegneria dell'Informazione e modelli Matematici \endgraf
Universit\'a degli Studi di Palermo\endgraf
Viale delle Scienze, I-90128\endgraf
Palermo, Italy\endgraf
and\endgraf
Istituto Nazionale di Fisica Nucleare \endgraf
via Cinthia, Montesantangelo\endgraf
Napoli, Italy\endgraf}
\email{fabio.bagarello@unipa.it}

\author[F.G. Russo]{Francesco G. Russo}
\address{Department of Mathematics and Applied Mathematics\endgraf
 University of Cape Town \endgraf
 Private Bag X1, Rondebosch 7701, Cape Town, South Africa\endgraf}
\email{francescog.russo@yahoo.com}


\date{\today}


\begin{abstract}
\noindent We show how the one-mode pseudo-bosonic ladder operators provide concrete examples of  nilpotent Lie algebras of dimension five. It is the first time that  an algebraic-geometric structure of this kind is observed in the context of pseudo-bosonic operators. Indeed we don't find the well known Heisenberg algebras, which are involved in several quantum dynamical systems,  but different Lie algebras which may be decomposed in the sum of two abelian Lie algebras in a prescribed way. We introduce the notion of semidirect sum (of Lie algebras) for this scope and find that it describes  very well  the behaviour of pseudo-bosonic operators in many quantum models.
\end{abstract}

\date{\today}

\maketitle

\section{Introduction}

There are various fields of geometry and physics which are connected since a long time and the notions of Heisenberg algebra and of Heisenberg group appear often in the description of some dynamical systems in which the use of the quantum mechanics is significant. The classical reference of Messiah \cite{messiah} shows the importance of the overlaps  of Functional analysis and the theory of Lie algebras in  quantum mechanics.   Heisenberg algebras (and Heisenberg groups) are often involved in several  aspects of theoretical physics, for instance in any situation in which symmetries are present. For this reason, it is useful to formalize the well known notion of nilpotence in an algebraic structure and adapt it to physical contexts (see \cite{snobl}).

Heisenberg algebras are a special  class of finite dimensional nilpotent Lie algebras, which have been historically important in theoretical physics. On the other hand, they are not alone and the classification of nilpotent Lie algebras of finite dimension (see \cite{anch, degraaf, gong, morozov, turk}) shows that  it is possible to find several other Lie algebras of small dimension with significant applications in physics (see again \cite{snobl}). The classifications by Beck and Kolman \cite{beckekolman}  of 6-dimensional nilpotent Lie
algebras over fields of characteristics zero overlap a series of important contributions, begun with Anchoch\'ea-Bermudez \cite{anch},  Morozov \cite{morozov}, Turkowski \cite{turk} and Gong \cite{gong}. These works  are  improved by de Graaf  \cite{degraaf} to arbitrary fields of arbitrary characteristic. We will use their classifications in order to describe mathematical models in which pseudo-bosonic operators (see Sections 3 and 4) are involved. In Section 2 we give some useful definitions and notions of theoretical nature on Lie algebras and their decompositions. These notions can be found in \cite{goze,knapp} and are recently used in \cite{nr1, nr2, nr3}. Then we present the essential elements of the theory of the pseudo-bosonic operators as developed in \cite{baginbagbook} and references therein, where we observe that it is possible to introduce a Lie bracket, and so the structure of Lie algebra, on functional spaces of pseudo-bosonic operators.
The relations between functional analysis, operator algebras and Lie algebra structures of bounded operators are recently discussed in \cite{beltita}, which is an useful reference on the connections of these two subjects. But it is important to stress that pseudo-bosonic operators are unbounded, so that their treatment does not fit in the framework considered in \cite{beltita}. For this reason, we devote  Sections 3 and 4  to explore the above mentioned connection in the context of unbounded, pseudo-bosonic operators, so we have all that we need for the proofs of Section 5, where the main results are placed. Finally, we describe a new perspective in terms of Schur multipliers, providing a computation for a specific Lie algebra which is involved in our main results. Notations and terminology are standard and follow the main references.

\section{Preliminaries of geometric nature on Lie algebras}

For future convenience, we introduce some of the notions which will play a role in the paper.

\begin{defn}\label{semidirect-def}
A Lie algebra $\mathfrak{l}$ is the {\em semidirect sum} of two of its Lie subalgebras $\mathfrak{a}$ and $\mathfrak{b}$ if  the following conditions are satisfied:
\begin{itemize}
\item[(i)]$\mathfrak{a}$  is an ideal of $\mathfrak{l}$,
\item[(ii)] $\mathfrak{l}= \mathfrak{a} + \mathfrak{b}$,
\item[(iii)] $\mathfrak{a} \cap \mathfrak{b}=0$.
\end{itemize}
In this situation, $\mathfrak{a}$ and $\mathfrak{b}$ are called \textit{factors} of $\mathfrak{l}$.
\end{defn}

One of the first examples of a Lie algebra satisfying the conditions of Definition \ref{semidirect-def} is offered by the direct sum $\mathfrak{a} \oplus \mathfrak{a} $ of two copies of an abelian Lie algebra $\mathfrak{a}$.  More generally, we may consider an abelian Lie algebra $\mathfrak{b}$ different from $\mathfrak{a}$ and again the direct sum $\mathfrak{a} \oplus \mathfrak{b} $ satisfies the conditions of Definition \ref{semidirect-def}. One can argue with infinite families of abelian Lie algebras, not restricting  the case at  two, and in  \cite{knapp} one can find details of the following example.

\begin{ex} \label{directproducts} Given an abelian Lie algebra $\mathfrak{a}_n$ and  $n \ge 1$, the direct sum $$\mathfrak{l}=\mathfrak{a}_1 \oplus \mathfrak{a}_2 \oplus \ldots \oplus \mathfrak{a}_n \oplus \mathfrak{a}_{n+1} \oplus \ldots $$
 of countably many abelian Lie algebras $\mathfrak{a}_n$  is a Lie algebra satisfying  Definition \ref{semidirect-def}. In fact, if $\mathfrak{a}=\mathfrak{a}_1$ and $\mathfrak{b}=\mathfrak{a}_2 \oplus \ldots \oplus \mathfrak{a}_n \oplus \mathfrak{a}_{n+1} \oplus  \ldots$, then we may easily check that (i), (ii) and (iii) of Definition \ref{semidirect-def} are satisfied. In the present example, we have the additional condition  $[\mathfrak{b},\mathfrak{l}] \subseteq \mathfrak{b}$, that is,  $\mathfrak{b}$ is an ideal of  $\mathfrak{l}$ that is not required in Definition \ref{semidirect-def}. The same conclusions can also be deduced, if $\mathfrak{l}$ is constructed as above, but $\mathfrak{a}_n$ is not abelian (see \cite{goze, knapp}).
\end{ex}

Indeed all direct sums of Lie algebras are particular cases of the notion of semidirect sum, once an additional condition is satisfied together with (i), (ii) and (iii) of Definition \ref{semidirect-def}.

\begin{prop}\label{fact}
Assume $\mathfrak{l}=\mathfrak{a} + \mathfrak{b}$ is a Lie algebra satisfying  Definition \ref{semidirect-def}. If we have in addition that $\mathfrak{b}$ is an ideal of $\mathfrak{l}$, then $\mathfrak{l}=\mathfrak{a} \oplus \mathfrak{b}$.
\end{prop}

\begin{proof}
See \cite[Chapter 4, \S 2]{knapp} for the proof.
\end{proof}

The reader can find the corresponding notion of \textit{semidirect product} of two groups in \cite[Definition 5.72]{hofmor} and Proposition \ref{fact} is described by \cite[Exercises E5.25, E5.26, E5.27]{hofmor} in terms of groups, instead of Lie algebras. This is to show that we are recalling some well known  constructions in algebra and geometry.

At this point, it is interesting to look for Lie algebras which are not necessarily of the forms that we have listed until now. One of the interesting examples, which will play a fundamental role in our studies, is the \textit{Heisenberg Lie algebras}, or briefly \textit{Heisenberg algebras}. We need to recall some classical notions from \cite{ goze, hofmor, knapp}, in order to define Heisenberg algebras.

\begin{defn} \label{derived}If $\mathfrak{l}$ is a Lie algebra (over the field $\mathbb{C}$ of the complex numbers),  the \textit{derived subalgebra} of $\mathfrak{l}$ is the smallest subalgebra
\[[\mathfrak{l},\mathfrak{l}]=\langle [a,b]  \ | \ a,b \in \mathfrak{l} \rangle \]
of $\mathfrak{l}$ generated by the commutators $[a,b]=ab-ba$.
\end{defn}

From the literature (see \cite{goze, hofmor, knapp}), it is well known that $[\mathfrak{l},\mathfrak{l}]$ is  an ideal of $\mathfrak{l}$ and that the quotient  $$\frac{\mathfrak{l}}{[\mathfrak{l},\mathfrak{l}]}=\left\{ a+[\mathfrak{l},\mathfrak{l}] \ | \ a \in \mathfrak{l} \right\}$$ is an abelian Lie algebra.

Another important abelian subalgebra of $\mathfrak{l}$ is the following:

\begin{defn}\label{center}If $\mathfrak{l}$ is a Lie algebra (over the field $\mathbb{C}$ of the complex numbers),  the \textit{center} of $\mathfrak{l}$ is defined by
\[Z(\mathfrak{l})=\{ a \in \mathfrak{l}  \ | \ [a,b]=0  \ \forall b \in \mathfrak{l}\}. \]
 \end{defn}
The center of $\mathfrak{l}$ turns out to be, not only a set, but an ideal of $\mathfrak{l}$. In general we cannot determine whether the quotient Lie algebra $\mathfrak{l}/Z(\mathfrak{l})$ is  abelian  or not, but there are various criteria to determine how far is $\mathfrak{l}/Z(\mathfrak{l})$ from being abelian: one way is to define appropriate ideals and form a chain in such a way that one extends step by step the original center of $\mathfrak{l}$, as illustrated in \cite{goze, knapp}.

\begin{defn}\label{uppercentralseries}The \textit{upper central series} of $\mathfrak{l}$ is defined   inductively by
\[0=Z_0(\mathfrak{l}) \leq Z_1(\mathfrak{l})=Z(\mathfrak{l}) \leq Z_2(\mathfrak{l}) \leq \ldots \leq Z_i(\mathfrak{l}) \le Z_{i+1}(\mathfrak{l}) \leq \ldots, \]
where each $Z_i(\mathfrak{l})$ turns out to be an ideal of $\mathfrak{l}$ (called $i$th $center$ of $\mathfrak{l}$) and
\[ \frac{Z_1(\mathfrak{l})}{Z_0(\mathfrak{l})}=Z(\mathfrak{l}), \frac{Z_2(\mathfrak{l})}{Z_1(\mathfrak{l})}=Z\left(\frac{\mathfrak{l}}{Z(\mathfrak{l})}\right),   \ldots, \frac{Z_{i+1}(\mathfrak{l})}{Z_i(\mathfrak{l})}=Z\left(\frac{\mathfrak{l}}{Z_i(\mathfrak{l})}\right), \ldots. \]
\end{defn}

Using elementary computations, one can see that the above definition of $i$th center of $\mathfrak{l}$ (in terms of quotients) is equivalent (in terms of commutators) to
\[ Z_{i+1}(\mathfrak{l})=\{ a \in \mathfrak{l} \mid  [a,b]\in Z_i(\mathfrak{l}) \ \forall  b\in \mathfrak{l}  \}. \]
Let's give  a hint (omitting the first step which is trivial): since $Z_1(\mathfrak{l})=Z(\mathfrak{l})$, $Z(\mathfrak{l}/Z(\mathfrak{l}))$ is clearly abelian and so we have in terms of cosets $[a+Z(\mathfrak{l}),b+Z(\mathfrak{l})]=[a,b] + Z(\mathfrak{l}) = Z(\mathfrak{l})$, hence $[a,b] \in Z(\mathfrak{l})$. This gives the definition of  $Z_2(\mathfrak{l})$. Viceversa, beginning with the definition of $Z_2(\mathfrak{l})$, we get that $Z(\mathfrak{l}/Z(\mathfrak{l}))$ is abelian. Of course, the argument can be repeated step by step, in order to show what we claimed. At this point we may ask whether the upper central series of $\mathfrak{l}$ stops or not. There is a terminology for this scope.

\begin{defn}\label{nilpotent}
A Lie algebra $\mathfrak{l}$ is called \textit{nilpotent of class $c$} if $\mathfrak{l}=Z_c(\mathfrak{l})$, that is, if the upper central series of $\mathfrak{l}$ ends after a finite number $c \ge 0$ of steps.
\end{defn}

With the present terminology, it is clear that abelian Lie algebras are nilpotent  of class $c=1$. We are interested here to nilpotent Lie algebras of class $2$ and to their applications in the context of the models of mathematical physics developed in \cite{baginbagbook}-\cite{bag8}. One of the reasons is that  we find some important Lie algebras among nilpotent Lie algebras of class $2$; they appear since long time in the literature on quantum mechanics. Following \cite{goze, knapp, nr1, nr2},

\begin{defn}\label{Heisenberg}  A finite dimensional Lie algebra  $\mathfrak{l}$ is called $Heisenberg$ provided that $[\mathfrak{l},\mathfrak{l}]=Z(\mathfrak{l})$ and $\mathrm{dim}([\mathfrak{l},\mathfrak{l}]) = 1$. Such algebras are odd dimensional with
basis $v_1, \ldots , v_{2m}, v$ and the only non--zero commutator between basis elements is $[v_{2i-1}, v_{2i}] = -
[v_{2i}, v_{2i-1}]= v$ for $i = 1,2, \ldots ,m$. The symbol $\mathfrak{h}(m)$
denotes the Heisenberg algebra of dimension $2m + 1$.
\end{defn}

At this point, we come back to Definition  \ref{semidirect-def}.

\begin{prop}\label{semi-heisenberg}
There are two abelian Lie algebras $\mathfrak{a}$ and $\mathfrak{b}$ such that $\mathfrak{h}(m)$ is the semidirect sum of $\mathfrak{a}$ and $\mathfrak{b}$.
\end{prop}

\begin{proof} Considering $\mathfrak{a}=\langle v_1, v_3, v_5, \ldots, v_{2m-1}, v \rangle$ and $\mathfrak{b}=\langle v_2, v_4, \ldots, v_{2m} \rangle$, the commutator relations in $\mathfrak{h}(m)$ show that both $\mathfrak{a}$ and $\mathfrak{b}$ are abelian, they are disjoint and $\mathfrak{a}$ is an ideal of $\mathfrak{l}$. Moreover $\mathfrak{a} + \mathfrak{b}= \mathfrak{l}$. The result follows.
\end{proof}

In particular, it turns out to be useful the following consequence:

\begin{cor}\label{trick} If $\mathfrak{c}$ is an abelian Lie algebra, then  $\mathfrak{l}=\mathfrak{h}(m) \oplus \mathfrak{c}$  is the semidirect sum of two abelian Lie algebras $\mathfrak{d}$ and $\mathfrak{e}$.
\end{cor}

\begin{proof} Define $\mathfrak{a}$ and $\mathfrak{b}$ as in Proposition \ref{semi-heisenberg}. Then
$\mathfrak{h}(m)= \mathfrak{a} + \mathfrak{b}$ and so
$$\mathfrak{l}=(\mathfrak{a} + \mathfrak{b}) \oplus \mathfrak{c} =  \mathfrak{a} + (\mathfrak{b} \oplus \mathfrak{c}), $$
where we get the result for $\mathfrak{d}=\mathfrak{a}$ and $\mathfrak{e}=(\mathfrak{b} \oplus \mathfrak{c})$, abelian by Example \ref{directproducts}.
\end{proof}

The reader may note that Proposition \ref{semi-heisenberg} is even providing an additional information on the factorization  of $\mathfrak{h}(m)$ in the semidirect sum of two Lie algebras which are abelian. $\mathfrak{h}(m)$ is particularly relevant when we want to look for classifications of finite dimensional nilpotent Lie algebras over the field of the complex numbers. Indeed we end this section with an important result of classification (see \cite[p.646]{degraaf}), which will be used in the main results of the present paper.

\begin{thm}[Classification of Finite Dimensional Nilpotent Lie Algebras of Dimension 3, 4 and  5]\label{classification}
Let $\mathfrak{l}$ be a finite dimensional nilpotent Lie algebra (over any field of characteristic $\neq 2$) and $\mathfrak{i}$ an abelian Lie algebra of dimension 1. Then
\begin{itemize}
\item[(1)] $\mathrm{dim} \ \mathfrak{l} =3$ if and only if $\mathfrak{l}$ is isomorphic to one of the following Lie algebras:
 \begin{itemize}
\item[-] $\mathfrak{l}_{3,1}$ abelian Lie algebra of dimension 3,
\item[-] $\mathfrak{l}_{3,2}  \simeq \mathfrak{h}(1)$.
\end{itemize}
\item[(2)]$\mathrm{dim} \ \mathfrak{l} =4$ if and only if    $\mathfrak{l}$ is isomorphic to one of the following Lie algebras:
 \begin{itemize}
\item[-] $\mathfrak{l}_{4,1} =  \mathfrak{l}_{3,1} \oplus \mathfrak{i}$,
\item[-] $\mathfrak{l}_{4,2} =  \mathfrak{l}_{3,2} \oplus \mathfrak{i}$
\item[-] $\mathfrak{l}_{4,3} = \langle v_1, v_2, v_3, v_4 \ | \ [v_1,v_2]=v_3, [v_1,v_3]=v_4 \rangle$
\end{itemize}
\item[(3)]$\mathrm{dim} \ \mathfrak{l} =5$ if and only if $\mathfrak{l}$ is isomorphic to one of the following Lie algebras:
\begin{itemize}
\item[-] $\mathfrak{l}_{5,1} = \mathfrak{l}_{4,1} \oplus \mathfrak{i} $,
\item[-] $\mathfrak{l}_{5,2} = \mathfrak{l}_{4,2} \oplus \mathfrak{i} \simeq \mathfrak{h}(1) \oplus \mathfrak{i} \oplus \mathfrak{i}$,
\item[-] $\mathfrak{l}_{5,3} = \mathfrak{l}_{4,3} \oplus \mathfrak{i} $,
\item[-] $\mathfrak{l}_{5,4} =  \langle v_1, v_2, v_3, v_4, v_5 \ | \ [v_1,v_2]=[v_3,v_4]=v_5 \rangle \simeq \mathfrak{h}(2) $,
\item[-] $\mathfrak{l}_{5,5} =  \langle v_1, v_2, v_3, v_4, v_5 \ | \ [v_1,v_2]=v_3, [v_1,v_3]=[v_2,v_4]=v_5 \rangle$,
\item[-] $\mathfrak{l}_{5,6} =  \langle v_1, v_2, v_3, v_4, v_5 \ | \ [v_1,v_2]=v_3, [v_1,v_3]=v_4,$ $$[v_1,v_4]=[v_2,,v_3]=v_5  \rangle,$$
\item[-] $\mathfrak{l}_{5,7} =  \langle v_1, v_2, v_3, v_4, v_5 \ | \ [v_1,v_2]=v_3, [v_1,v_3]=v_4, [v_1,v_4]=v_5 \rangle$,
\item[-] $\mathfrak{l}_{5,8} =  \langle v_1, v_2, v_3, v_4, v_5 \ | \ [v_1,v_2]=v_4, [v_1,v_3]=v_5 \rangle$,
\item[-] $\mathfrak{l}_{5,9} =  \langle v_1, v_2, v_3, v_4, v_5 \ | \ [v_1,v_2]=v_3, [v_1,v_3]=v_4, [v_2,v_3]=v_5 \rangle$.
\end{itemize}
\end{itemize}
\end{thm}

In dimension five, the reader can note the presence of  decomposable Lie algebras in Theorem \ref{classification}, that is, Lie algebras satisfying Definition \ref{semidirect-def} with the presence of an abelian factor $\mathfrak{i}$ of dimension one.  For instance, these are $\mathfrak{l}_{5,k}$, when $k=1,2,3$, while it is more difficult to see whether  $\mathfrak{l}_{5,k}$ satisfies, or does not satisfy, the conditions of Definition \ref{semidirect-def} only looking at the generators and at the relations. We will recognize concrete physical models in which this is possible and will give details later on, identifying generators and relations in the case of nilpotent Lie algebras of dimension 5 via functional operators which are fundamental in quantum mechanics.

\section{A short review on $\mathcal{D}$-pseudo bosons}\label{sectpbs}

The next two sections are devoted to recall some notions of functional analysis and operator algebras, which can be found in \cite{aitbook, chri, schu,baginbagbook}, which are relevant for us.  Let $\mathcal{H}$ be a given Hilbert space (over the field $\mathbb{C}$ of complex numbers) with scalar product $\left<.,.\right>$ and related norm $\|.\|$.
Let $a$ and $b$ be two operators
on $\mathcal{H}$, with domains $D(a)$ and $D(b)$ respectively, $a^\dagger$ and $b^\dagger$ their adjoint, and let $\mathcal{D}$ be a dense subspace of $\mathcal{H}$
such that $a^\sharp\mathcal{D}\subseteq\mathcal{D}$ and $b^\sharp \mathcal{D} \subseteq \mathcal{D}$, where with $x^\sharp$ we indicate $x$ or $x^\dagger$. Of course, $\mathcal{D}\subseteq D(a^\sharp)$
and $\mathcal{D}\subseteq D(b^\sharp)$.

\begin{defn}\label{def21}
The operators $(a,b)$ are $\mathcal{D}$-\textit{pseudo-bosonic}  if, for all $f\in\mathcal{D}$, we have
\begin{equation}\label{A1}
a\,b\,f-b\,a\,f=f.
\end{equation}
\end{defn}

\begin{rem} Notice that, in the particular case when $b=a^\dagger$, the pseudo-bosonic operators are just  creation and annihilation operators obeying the \textit{canonical commutation relations} (CCR), $[c,c^\dagger]=\mathbb{I}$, which are well known to be unbounded. Hence they can only be defined on suitable domains. We meet similar problems when the CCR are replaced by \eqref{A1}.
\end{rem}

What we want to do now is to extend the ordinary construction well known for the CCR  to our case. For bosons we know that a vacuum $e_0$ does exist which is annihilated by $c$, $c\,e_0=0$, and which belongs to the domain of all the powers of $c^\dagger$. For each $n \ge 0$,  we can construct a set of vectors of $\mathcal{H}$
$$e_n=\frac{1}{\sqrt{n!}}(c^\dagger)^n e_0$$
 which are all in the domain of the number operator $N_0=c^\dagger c$, and which are its eigenstates:  $N_0e_n=ne_n$. The set
 $$\mathcal{F}_e=\{e_n \ | \ n\geq 0\}$$ is in fact an orthonormal basis for $\mathcal{H}$. If $C^\infty(\mathbb{R})$ denotes the smooth functions with support on $\mathbb{R}$ and ${\| \  \|}_2$ the usual $L^2$-norm, we fix $$\mathcal{H}={\mathcal L}^2(\Bbb R)=\{f(x) \mbox{ is Lebesgue-measurable on } \Bbb R \ | \ {\| f \|}_2 < \infty \}$$ then each $e_n$ is a well known function $e_n(x)$ in
\begin{equation}\mathcal{S}(\Bbb R)=\{f(x) \in C^\infty(\mathbb{R})  \  |  \  \lim_{|x|,\infty}|x|^kf^{(l)}(x)=0, \, \forall k,l\geq0  \ \},\label{sr}\end{equation}
 the set of  smooth functions which decrease, together with their derivatives, faster than any inverse power (see \cite{messiah}).

When CCR are replaced by (\ref{A1}), there is no a priori reason for such a situation to remain unchanged. For this reason, we need to impose some reasonable conditions which are verified in explicit models, and which reproduce back the well known bosonic settings when $b=a^\dagger$. In particular, our starting assumptions are the following:

\vspace{2mm}

{\bf Assumption $\mathcal{D}$-pb 1.--}  there exists a non-zero $\varphi_{ 0}\in\mathcal{D}$ such that $a\,\varphi_{ 0}=0$.

\vspace{1mm}

{\bf Assumption $\mathcal{D}$-pb 2.--}  there exists a non-zero $\Psi_{ 0}\in\mathcal{D}$ such that $b^\dagger\,\Psi_{ 0}=0$.

\vspace{2mm}

It is obvious that, since $\mathcal{D}$ is stable under the action of the operators introduced above,  $\varphi_0\in D^\infty(b)=\cap_{k\geq0}D(b^k)$ and  $\Psi_0\in D^\infty(a^\dagger)$, so
that the vectors
\begin{equation}\label{A2}
 \varphi_n=\frac{1}{\sqrt{n!}}\,b^n\varphi_0,\qquad \Psi_n=\frac{1}{\sqrt{n!}}\,{a^\dagger}^n\Psi_0,
 \end{equation}
$n\geq0$, can be defined and they all belong to $\mathcal{D}$. Then, they also belong to the domains of $a^\sharp$, $b^\sharp$ and $N^\sharp$, where $N=ba$.

We see that, from a practical point of view, $\mathcal{D}$ is the natural space to work with and, in this sense, it is even more relevant than $\mathcal{H}$. Let's put $$\mathcal{F}_\Psi=\{\Psi_{ n} \ |  \ n\geq0\} \ \mathrm{and} \  \mathcal{F}_\varphi=\{\varphi_{ n} \ | \ n\geq0\}.$$
It is  simple to deduce the following lowering and raising relations:
\begin{equation}\label{A3}
\left\{
    \begin{array}{ll}
b\,\varphi_n=\sqrt{n+1}\varphi_{n+1}, \qquad\qquad\quad\,\, n\geq 0,\\
a\,\varphi_0=0,\quad a\varphi_n=\sqrt{n}\,\varphi_{n-1}, \qquad\,\, n\geq 1,\\
a^\dagger\Psi_n=\sqrt{n+1}\Psi_{n+1}, \qquad\qquad\quad\, n\geq 0,\\
b^\dagger\Psi_0=0,\quad b^\dagger\Psi_n=\sqrt{n}\,\Psi_{n-1}, \qquad n\geq 1,\\
       \end{array}
        \right.
\end{equation}
 as well as the eigenvalue equations $N\varphi_n=n\varphi_n$ and  $N^\dagger\Psi_n=n\Psi_n$, $n\geq0$. In particular, as a consequence
of these last two equations,  if we choose the normalization of $\varphi_0$ and $\Psi_0$ in such a way that $\left<\varphi_0,\Psi_0\right>=1$, then we may deduce
\begin{equation} \label{A4} \left<\varphi_n,\Psi_m\right>=\delta_{n,m},
\end{equation}
 for all $n, m\geq0$. Hence $\mathcal{F}_\Psi$ and $\mathcal{F}_\varphi$ are biorthogonal.

 It is easy to see that, if $b=a^\dagger$, then $\varphi_n=\Psi_n=e_n$ (identifying $a$ with $c$), so that biorthogonality is replaced by a simpler orthonormality. Moreover, the relations in (\ref{A3}) collapse, and only one number operator exists, since in this case $N=N^\dagger$.

 The analogy with ordinary bosons suggests us to consider the following:

\vspace{2mm}

{\bf Assumption $\mathcal{D}$-pb 3.--}  $\mathcal{F}_\varphi$ is a basis for $\mathcal{H}$.

\vspace{1mm}

This is equivalent to requiring that $\mathcal{F}_\Psi$ is a basis for $\mathcal{H}$ as well. The reader may refer to \cite{chri} for more details. However, several  physical models show that $\mathcal{F}_\varphi$ is {\bf not}, in general, a basis for $\mathcal{H}$, but it is still complete (or total, as some authors prefer to say) in $\mathcal{H}$. The present circumstance suggests to introduce a weaker version of  Assumption $\mathcal{D}$-pb 3, which was originally studied in \cite{baginbagbook}. In order to do this, we recall the following notion:

\begin{defn} Let $\mathcal{G}$ be a dense subspace in a Hilbert space $\mathcal{H}$. For all $f$ and $g$ in $\mathcal{G}$,
we say that $\mathcal{F}_\varphi$ and $\mathcal{F}_\Psi$ are $\mathcal{G}$-\textit{quasibases} for $\mathcal{H}$, if
\begin{equation}\label{A4b}
\left<f,g\right>=\sum_{n\geq0}\left<f,\varphi_n\right>\left<\Psi_n,g\right>=\sum_{n\geq0}\left<f,\Psi_n\right>\left<\varphi_n,g\right>.
\end{equation}
\end{defn}

This notion can be seen as a weak form of the resolution of the identity, restricted to $\mathcal{G}$. Of course, if $f\in\mathcal{G}$ is orthogonal to all the $\varphi_n$'s, or to all the $\Psi_n$'s, then (\ref{A4b}) implies that $f=0$. Hence $\mathcal{F}_\varphi$ and $\mathcal{F}_\Psi$ are complete in $\mathcal{G}$.

We have all we need, in order to formulate properly the following condition.

\vspace{2mm}

{\bf Assumption $\mathcal{D}$-pbw 3.--}  For some subspace $\mathcal{G}$ dense in $\mathcal{H}$, $\mathcal{F}_\varphi$ and $\mathcal{F}_\Psi$ are $\mathcal{G}$-quasi bases.

\vspace{2mm}

On the other hand, in some particular case, it is possible to replace Assumption $\mathcal{D}$-pbw 3 above with its stronger version:

\vspace{2mm}

{\bf Assumption $\mathcal{D}$-pbs 3.--}  $\mathcal{F}_\varphi$ is a Riesz basis for $\mathcal{H}$.

\vspace{1mm}

We recall from \cite{chri} that  $\mathcal{F}_\varphi$ is a \textit{Riesz basis} for $\mathcal{H}$, if a bounded operator $S$, with bounded inverse $S^{-1}$, exists in $\mathcal{H}$, together with an orthonormal basis $\mathcal{F}_{\hat e}=\{\hat e_n\in\mathcal{H} \ | \ n\geq0\}$ such that $\varphi_n=S\hat e_n$, for all $n\geq0$.  Note that the uniqueness of the  basis biorthogonal to $\mathcal{F}_\varphi$, implies automatically that $\mathcal{F}_\Psi$ is also a Riesz basis for $\mathcal{H}$. In particular, we have that $\Psi_n=(S^{-1})^\dagger \hat e_n$.

$\mathcal{D}$ pseudo-bosons appear in several physical models arising in PT-quantum mechanics, but not only. Among the others,  it is possible to write in terms of $\mathcal{D}$ pseudo-bosonic operators $a$ and $b$:
\begin{itemize}
\item[(1)] the Hamiltonian for the  extended quantum harmonic oscillator (see \cite{dapro,bagswans});
\item[(2)] the Hamiltonian for the Swanson  model (see \cite{swans,bagswans});
\item[(3)] the Hamiltonian for the Black and Scholes equation (see \cite{roy2,bag2});
\item[(4)] the Hamiltonian for the models proposed by Bender and Jones (see \cite{benjon,baglat}).
\end{itemize}
We refer to \cite{baginbagbook} for several applications to physics of this framework, while more recent applications can be found in \cite{bagPR2015,bagrevA}. We also refer to \cite{baginbagbook} for more results on different aspects of $\mathcal{D}$-PBs.

Here, in view of our interest to Lie algebras, we are more interested in introducing an algebraic framework in which formula (\ref{A1}) can be cast in the form of a Lie bracket. In other words, we are interested to introduce now an algebraic rather than a functional point of view for our pseudo-bosonic operators. This will be done in the next section.

\section{An algebraic perspective  for the operators $a$ and $b$}

As we have discussed before, it is quite reasonable to expect that $a$, $b$ and their adjoints are not bounded operators. The simplest way to understand it is just to consider the case in which $b=a^\dagger$. In fact, in this case, it is well known that the ladder operators, and the related number operator as a consequence, are unbounded. When $b\neq a^\dagger$ the situation changes a little bit, but it is still possible to see, \cite{baginbagbook}, that in concrete systems these operators are still unbounded. This means that we cannot simply multiply, say, $a$ and $b$, since the product of unbounded operators needs not to be defined. However, in recent years, a lot of work on unbounded operator algebras has been carried out by several authors, \cite{aitbook,bagrev2007,schu,trrev}. The essential idea in all these contributions is to construct an useful framework in which the product of some unbounded operators can be defined, and the result of this product returns an element of the same structure.

There exist several possible approaches to this problem, motivated by the theory of the so called $*$-$algebras$. We invite the reader to look at \cite{aitbook} for the notion of $*$-algebra. In particular we can construct what is called a {\em quasi $*$-algebras}, (see \cite{aitbook, trrev}), or more concrete $O^*$-$algebras$ (see \cite{bagrev2007}). This is what we will do now, since it seems more useful for our particular problem. The notion of $O^*$-algebra is recalled below:

\begin{defn}\label{o*}Let $\mathcal{H}$ be a separable Hilbert space and $N_0$ an
unbounded, densely defined, self-adjoint operator. Let $D(N_0^k)$ be
the domain of the operator $N_0^k$, $k \ge 0$, and $\mathcal{D}$ the domain of
all the powers of $N_0$, that is,  $$ \mathcal{D} = D^\infty(N_0) = \bigcap_{k\geq 0}
D(N_0^k). $$ This set is dense in $\mathcal{H}$. Let us now introduce
$\mathcal{L}^\dagger(\mathcal{D})$, the $*$-algebra of all the \textit{  closable operators}
defined on $\mathcal{D}$ which, together with their adjoints, map $\mathcal{D}$ into
itself. Here the adjoint of $X\in\mathcal{L}^\dagger(\mathcal{D})$ is
$X^\dagger=X^*_{| \mathcal{D}}$. Such an $\mathcal{L}^\dagger(\mathcal{D})$ is called  $O^*$-algebra.
\end{defn}

Note that in the above definition one needs to introduce a map
which, given an element $X\in\mathcal{L}^\dagger(\mathcal{D})$, produces another
element $X^\dagger \in \mathcal{L}^\dagger(\mathcal{D})$. The most natural choice,
which is clearly $X^\dagger\equiv X^*$, is  only compatible with
$\mathcal{L}^\dagger(\mathcal{D})=B(\mathcal{H})$, i.e. with $N_0$ bounded, which is not what
we want. Recalling that $D(X^*)\supseteq\mathcal{D}$, it is clear that
$X^*_{| \mathcal{D}}$ is well defined. Further one can prove that
$\dagger$ has the properties of an involution and maps
$\mathcal{L}^\dagger (\mathcal{D})$ into itself.

In $\mathcal{D}$ the topology is defined by the following $N_0$-depending
seminorms: $$\phi \in \mathcal{D} \rightarrow \|\phi\|_n\equiv \|N_0^n\phi\|,$$
where $n \ge 0$, and  the topology $\tau_0$ in $\mathcal{L}^\dagger(\mathcal{D})$ is introduced by the seminorms
$$ X\in \mathcal{L}^\dagger(\mathcal{D}) \rightarrow \|X\|^{f,k} \equiv
\max\left\{\|f(N_0)XN_0^k\|,\|N_0^kXf(N_0)\|\right\},$$ where
$k \ge 0$ and   $f \in \mathcal{C}$, the set of all the positive,
bounded and continuous functions  on $\mathbb{R}_+$, which are
decreasing faster than any inverse power of $x$:
$\mathcal{L}^\dagger(\mathcal{D})[\tau_0]$ is a {   complete *-algebra}.

As a consequence, if $x,y\in \mathcal{L}^\dagger(\mathcal{D})$, we can multiply the two elements and the results, $xy$ and $yx$, both belong to $\mathcal{L}^\dagger(\mathcal{D})$. This is what we were looking for: a suitable structure in which we have the possibility of introducing Lie brackets for objects (linear operators, in our case), which are not everywhere defined.
In fact, for each pair $x,y\in\mathcal{L}^\dagger(\mathcal{D})$, we can define a map $[.,.]$ as follows:
\begin{equation}\label{B1}
[x,y]=xy-yx.
\end{equation}

It is clear that $[x,y]\in\mathcal{L}^\dagger(\mathcal{D})$, that it is bilinear, that $[x,x]=0$ for all $x\in\mathcal{L}^\dagger(\mathcal{D})$, and that it satisfies the Jacobi identity
$$
[x,[y,z]]+[y,[z,x]]+[z,[x,y]]=0,
$$
for all $x,y,z\in\mathcal{L}^\dagger(\mathcal{D})$. Therefore, $[.,.]$ is a Lie bracket defined on $\mathcal{L}^\dagger(\mathcal{D})$.

\begin{rem} It is good to note here that \eqref{B1} is exactly a generator of the derived subalgebra $[\mathcal{L}^\dagger(\mathcal{D}),\mathcal{L}^\dagger(\mathcal{D})]$, introduced in a pure algebraic context in Definition \ref{derived} for an arbitrary Lie algebra. We point out this fact, because it is useful to see how algebraic notions and notions of functional analysis are strongly related.
\end{rem}

For our purposes it is now convenient to identify $\mathcal{D}$ with the set $\mathcal{S}(\Bbb R)$ of the test functions introduced in (\ref{sr}).

Collecting what we discussed and proved until now, we get the following {\em constructing} result.

\begin{prop}[Construction]\

\begin{itemize}

\item[(1)] We may introduce the self-adjoint position and momentum operators, $x$ and $p$, satisfying the commutation rule $[x,p]=i\mathbb{I}$;

\item[(2)] Using $x$ and $p$, we may define $N_0$ (essentially) as the Hamiltonian of the quantum harmonic oscillator: $N_0=p^2+x^2$, as well as their linear combinations operators $c=\frac{1}{\sqrt{2}}(x+ip)$ and $c^\dagger=\frac{1}{\sqrt{2}}(x-ip)$;

\item[(3)] We construct the algebra $\mathcal{L}^\dagger(\mathcal{D})$ as in Definition \ref{o*};

\item[(4)] We observe that $c, c^\dagger\in\mathcal{L}^\dagger(\mathcal{D})$, together with $N_0=2c^\dagger c+\mathbb{I}$;

\item[(5)] We note that $c$ and $c^\dagger$ are unbounded, as well as $N_0$. Therefore, the product of $N_0$ and $c$ makes no sense in $B(\mathcal{H})$, the algebra of bounded operators on $\mathcal{H}$, but it is perfectly defined in $\mathcal{L}^\dagger(\mathcal{D})$;

\item[(6)]We note that the topology on $\mathcal{D}$ in  Definition \ref{o*} coincides with the standard topology $\tau_\mathcal{S}$ on $\mathcal{S}(\Bbb R)$  in  \eqref{sr}.

\end{itemize}
\end{prop}

Let now $T$ be an invertible operator on $\mathcal{H}$ such that $T,T^{-1}\in\mathcal{L}^\dagger(\mathcal{D})$. This means, in particular, that $\mathcal{D}$ is stable under the action of $T$, $T^{-1}$, and of their adjoints. Hence we can introduce two operators, similar to $c$ and $c^\dagger$, as follows:
\begin{equation}\label{B2}
a=TcT^{-1},\qquad b=Tc^\dagger T^{-1}.
\end{equation}
It is clear that, because of our assumption on $T$, $a,b\in \mathcal{L}^\dagger(\mathcal{D})$ as well. The adjoints of $a$ and $b$ can be easily computed,
$$
a^\dagger=(T^{-1})^\dagger c^\dagger T^\dagger, \qquad b^\dagger=(T^{-1})^\dagger c T^\dagger,
$$
and they both belong to $\mathcal{L}^\dagger(\mathcal{D})$ again. Moreover, it is easy to see that both the operators $a$ and $b$ satisfy Definition \ref{def21}: they both map, with their adjoints, $\mathcal{D}$ into $\mathcal{D}$ and, taken any $f\in\mathcal{D}$, $abf-baf=f$. It is also clear that Assumptions $\mathcal{D}$-pb 1. and $\mathcal{D}$-pb 2. are both satisfied. In fact, the function $e_0(x)=\frac{1}{\pi^{1/4}}\,e^{-x^2/2}$ belongs to $\mathcal{S}(\Bbb R)$ and is annihilated by $c$. Since $T$ maps $\mathcal{S}(\Bbb R)$ into itself, and since it is invertible, it is clear that the nonzero function $\varphi_0(x)=Te_0(x)$ belongs to $\mathcal{S}(\Bbb R)$ and is annihilated by $a$. Analogously the nonzero function $\Psi_0(x)=(T^\dagger)^{-1}e_0(x)$ also belongs to $\mathcal{S}(\Bbb R)$ and is annihilated by $b^\dagger$.

We can then construct the functions $\varphi_n(x)$ and $\Psi_n(x)$ as we have shown before, and they all belong to $\mathcal{D}$. In general, under our assumption on $T$, we know that $\mathcal{F}_\varphi$ and $\mathcal{F}_\Psi$ are biorthogonal sets, but there is no guarantee they are bases, or Riesz bases. On the other hand, they can be easily proved to be $\mathcal{D}$-quasi bases, \cite{baginbagbook}. The conclusion is therefore that we have constructed pseudo-bosonic operators endowed by a Lie algebra structure. In the next section we will analyze more in detail this structure, considering some explicit examples.

\section{Main results}

In this section we will consider some examples of pseudo-bosonic operators introduced along the years in the literature, and we will discuss how these examples are related to some particular kind of Lie algebras. Interestingly enough, these Lie algebras will be shown to be of different kind for different physical systems, even if the functional structure is not very different in the various cases.

\begin{thm}[Decomposition of the Shifted Harmonic Oscillator]\label{main1} The pseudo-bosonic operator in \cite{bagrevA}, connected with the shifted harmonic oscillator, admits a Lie algebra structure $\mathfrak{a}_{sh}$ such that \begin{itemize}
\item[(i)] $\mathfrak{a}_{sh}= \mathfrak{a} + \mathfrak{b}$ is a semidirect sum of two abelian Lie algebras $\mathfrak{a}$ and $\mathfrak{b}$;
\item[(ii)] $\mathrm{dim} \ \mathfrak{a}_{sh} =5$, $\mathrm{dim} \ \mathfrak{a} =3$ and $\mathrm{dim} \ \mathfrak{b} =2$;
\item[(iii)] $[\mathfrak{a}_{sh},\mathfrak{a}_{sh}]=Z(\mathfrak{a}_{sh})$ has dimension 1 and $\mathfrak{a}_{sh}/Z(\mathfrak{a}_{sh})$ is an abelian Lie algebra of dimension $4$;
\item[(iv)] $\mathfrak{a}_{sh}$ and $\mathfrak{h}(2)$ are both non-abelian  Lie algebras of  class of nilpotence 2, but they are not isomorphic.
\item[(v)]$\mathfrak{a}_{sh}$ is isomorphic to $\mathfrak{l}_{5,2}$.
\end{itemize}
\end{thm}

\begin{proof}
Following the terminology and the notations we introduced in the previous sections, we may define
\begin{equation}\label{C1}
a=c-\alpha\mathbb{I},\qquad b=c^\dagger-\overline{\beta}\mathbb{I},
\end{equation}
where $\alpha$ and $\beta$ are different complex quantities such that $\alpha\neq\beta$. Hence $b\neq a^\dagger$, clearly. In \cite{bagrevA}, Section III,
it is shown that $a$ and $b$ can be written as in \eqref{B2}, for a suitable invertible operator $T$.

To check that $a$ and $b$ belong to $\mathcal{L}^\dagger(\mathcal{D})$, rather than checking if $T$ and $T^{-1}$ also belong to $\mathcal{L}^\dagger(\mathcal{D})$, it is much easier to consider the explicit form in \eqref{C1} of these operators: since $c,c^\dagger$ and $\mathbb{I}$ belong to $\mathcal{L}^\dagger(\mathcal{D})$, and since this is a linear vector space, it is clear that $a,b\in \mathcal{L}^\dagger(\mathcal{D})$ as well.  This argument could also be repeated for $a^\dagger$ and $b^\dagger$.

Hence we can use the general settings described in the previous section. For that we define $v_1=a$, $v_2=b$, $v_3=b^\dagger$, $v_4=a^\dagger$ and $v=\mathbb{I}$, and consider the Lie algebra.
\begin{equation}\label{C2}
\mathfrak{a}_{sh} = \langle v_1, v_2, v_3, v_4, v \  |  \  [v_1,v_2]=[v_3,v_4]=[v_1,v_4]=[v_3,v_2]=v,
\end{equation}
$$[v_1,v_3]=[v_2,v_4]=[v,v_j]=0 \quad \forall j=1,2,3,4\rangle.$$
Note that the Lie bracket, involved in the above presentation of $\mathfrak{a}_{sh}$, are exactly the rules which we have introduced, and commented, in \eqref{B1}.

We proceed to show (i). Define
\begin{equation}\label{AandB1}
\mathfrak{a} = \langle v_1,v_3,v \ | \ [v_1,v_3]=[v,v_1]= [v,v_3]=0 \rangle \ \mathrm{and} \ \mathfrak{b} = \langle v_2, v_4 \ | \ [v_2,v_4]=0 \rangle.
\end{equation}
Equivalently, $\mathfrak{a}= \langle a,b^\dagger,\mathbb{I} \ | \ [a,b^\dagger]=[a, \mathbb{I}]=[b^\dagger, \mathbb{I}]=0 \rangle$ and $\mathfrak{b} = \langle a^\dagger, b \ | \ [a^\dagger, b]=0 \rangle$.
Using the commutation rules in \eqref{C2}, it is easy to check that both $\mathfrak{a}$ and $\mathfrak{b}$ are Lie subalgebras of $\mathfrak{a}_{sh}$, that $\mathfrak{a}+\mathfrak{b}=\mathfrak{a}_{sh}$, and that $\mathfrak{a}\cap\mathfrak{b}=0$. Moreover, $[\mathfrak{a}, \mathfrak{a}_{sh}] \subseteq \mathfrak{a}$, so $\mathfrak{a}$ is an ideal of $\mathfrak{a}_{sh}$, while $[\mathfrak{b}, \mathfrak{a}_{sh}] \subseteq \langle v \rangle$ and $\langle v \rangle \not \subseteq \mathfrak{b}$, so   $\mathfrak{b}$ is not an ideal in $\mathfrak{a}_{sh}$. We may conclude that the conditions of Definition \ref{semidirect-def} are satisfied. In addition, the commutation relations show that both $\mathfrak{a}$ and $\mathfrak{b}$  are abelian. Therefore  (i) follows.

Statement (ii)  follows from \eqref{C2} and \eqref{AandB1}.

We proceed to show (iii). It is an easy computation to check that the center of $\mathfrak{a}_{sh}$, $Z(\mathfrak{a}_{sh})$, is only made by the multiples of $v$, and that the same is true for  $[\mathfrak{a}_{sh},\mathfrak{a}_{sh}]$, and so
\begin{equation}\label{C3}
Z(\mathfrak{a}_{sh})=[\mathfrak{a}_{sh},\mathfrak{a}_{sh}]=\{\lambda v,\,\lambda\in\Bbb C\} =  \langle v \rangle,
\end{equation}
is an abelian Lie algebra of dimension one. On the other hand, we may look at the upper central series of $\mathfrak{a}_{sh}$ and we discover that it stops after  two steps, because
$$\frac{\mathfrak{a}_{sh}}{Z(\mathfrak{a}_{sh})}=\langle v_1 + Z(\mathfrak{a}_{sh}), v_2 + Z(\mathfrak{a}_{sh}), v_3 + Z(\mathfrak{a}_{sh}), v_4 + Z(\mathfrak{a}_{sh}) \rangle $$
and $$[v_i + Z(\mathfrak{a}_{sh}), v_j + Z(\mathfrak{a}_{sh})] \subseteq Z(\mathfrak{a}_{sh}) \quad \forall i,j \in \{1,2,3,4\},$$
so $$\frac{\mathfrak{a}_{sh}}{Z(\mathfrak{a}_{sh})}=Z\left(\frac{\mathfrak{a}_{sh}}{Z(\mathfrak{a}_{sh})}\right)=\frac{Z_2(\mathfrak{a}_{sh})}{Z(\mathfrak{a}_{sh})}$$
is an abelian Lie algebra of dimension 4. Then (iii) follows.

We proceed to show (iv). Due to (ii) and Definition  \ref{Heisenberg}, if  $\mathfrak{a}_{sh}$   would be isomorphic to $\mathfrak{h}(m)$ for some $m$, then $m =2$ and so $\mathfrak{a}_{sh}$ would be isomorphic to  $\mathfrak{h}(2)$, but the presence of the nontrivial relation $[v_1,v_4]=v$ makes this impossible and so we may conclude that $\mathfrak{a}_{sh}$ cannot be isomorphic to $\mathfrak{h}(2)$. On the other hand, the reader may refer to \cite{nr1, nr2}, or do a direct computation via the rules in Definition \ref{Heisenberg}, in order to see that $\mathfrak{h}(2)$ is always a nonabelian  nilpotent Lie algebra of class of nilpotence 2.

In order to show (v), we may use \cite{gap} and check that the presentation of $\mathfrak{a}_{sh}$ is equivalent to that of $\mathfrak{l}_{5,2}$, involved in Theorem \ref{classification}.

\end{proof}

This is not the only situation in which a Lie algebra can be related to pseudo-bosons. In \cite{baglat} a two-dimensional model originally proposed in \cite{benjon} has been considered, showing that the Hamiltonian of the system can be written as a sum of two non self-adjoint operators proportional to two independent number-like operators of the kind introduced before. Here we show how a one-dimensional version of the same model fits well in our discussion, and in particular that the pseudo-bosonic operators produce a semidirect product of two algebras of the same kind of that we have discussed in Theorem \ref{main1}.

\begin{thm}[Decomposition of the Deformed Position and Momentum Operators]\label{main2}
The pseudo-bosonic operators in \cite{baglat} and \cite{benjon}   admit a Lie algebra structure $\mathfrak{b}$, which is isomorphic to the nonabelian nilpotent Lie algebra $\mathfrak{a}_{sh}$ of class 2 in \eqref{C2}.
\end{thm}

\begin{proof}Following \cite{baglat} we consider the self-adjoint operators $x$ and $p$ already introduced before, and we use them to introduce two new operators as follows: $P=p$ and $Q=x+i\alpha$, where $\alpha$ is a real constant. It is clear that $P=P^\dagger$, while $Q\neq Q^\dagger$. Also, the essential commutation rule is preserved: $[Q,P]=i\mathbb{I}$ (in the sense of unbounded operators). Then, for each fixed strictly positive $\beta$, we introduce the operators
\begin{equation}\label{C8}
a=\frac{1}{\sqrt{2\beta}}\left(iP+\beta Q\right), \qquad b=\frac{1}{\sqrt{2\beta}}\left(-iP+\beta Q\right),
\end{equation}
together with $a^\dagger$ and $b^\dagger$. In terms of $x$ and $p$, these operators can be written as follows:
$$
a=\frac{1}{\sqrt{2\beta}}\left(ip+\beta q+i\alpha\beta\right), \qquad b=\frac{1}{\sqrt{2\beta}}\left(-ip+\beta q+i\alpha\beta\right)
$$
and
$$
a^\dagger=\frac{1}{\sqrt{2\beta}}\left(-ip+\beta q-i\alpha\beta\right), \qquad b^\dagger=\frac{1}{\sqrt{2\beta}}\left(ip+\beta q-i\alpha\beta\right).
$$
Since $p=\frac{1}{\sqrt{2}\,i}(c-c^\dagger)$ and $q=\frac{1}{\sqrt{2}}(c+c^\dagger)$, all these operators can be written as linear combinations of $c$ and $c^\dagger$. Hence, they all belong to $\mathcal{L}^\dagger(\mathcal{D})$. Now, if we define $v_j$ and $v$ as in the proof of Theorem \ref{main1}, we easily recover that they satisfy the same Lie bracket as in \eqref{C2}. Hence, with the same arguments, we deduce that its algebra $\mathfrak{b}$ is isomorphic to $\mathfrak{a}_{sh}$. The result follows.
\end{proof}

We end with our final result, in which a different algebraic structure is involved. This model was originally considered in \cite{swans} and, later in \cite{bagswans}.

\begin{thm}[Decomposition of the Operators for Swanson model]\label{main3}
The pseudo-bosonic operators in \cite{swans}   admit a Lie algebra structure $\mathfrak{s}$ such that
\begin{itemize}
\item[(i)] $\mathfrak{s}= \mathfrak{a} + \mathfrak{b}$ is sum of two  Lie algebras $\mathfrak{a}$ and $\mathfrak{b}$;
\item[(ii)] The sum in (i) above is not semidirect;
\item[(iii)] $\mathfrak{a}$ is abelian and  $\mathfrak{b}$ is not abelian, but both are ideals of $\mathfrak{s}$;
\item[(iv)] $\mathrm{dim} \ \mathfrak{s}=5$, $\mathrm{dim} \ \mathfrak{a} =2$, $\mathrm{dim} \ \mathfrak{b} =4$ and $\mathrm{dim} \ \mathfrak{a} \cap   \mathfrak{b} =1$;
\item[(v)] $\mathfrak{s}$ is a nonabelian  nilpotent Lie algebra of class 2 with $[\mathfrak{s},\mathfrak{s}]=Z(\mathfrak{s})$ of dimension 1 and $\mathfrak{s}/Z(\mathfrak{s})$ is abelian of dimension $4$;
\item[(vi)] $\mathfrak{s}$ is isomorphic neither to $\mathfrak{h}(2)$ nor to $\mathfrak{a}_{sh}$.
\end{itemize}
\end{thm}

\begin{proof} Following \cite{swans, bagswans}, we may define operators $a$ and $b$ in terms of a third, purely bosonic, operator $c$ and of its adjoint $c^\dagger$ as follows:
\begin{equation}\label{C4}
a=\cos\theta\,c+i\sin\theta\, c^\dagger, \qquad b=\cos\theta\,c^\dagger+i\sin\theta\, c,
\end{equation}
so that their adjoints are
$$a^\dagger=\cos\theta\,c^\dagger-i\sin\theta\, c \qquad   \mathrm{and}  \qquad  b^\dagger=\cos\theta\,c-i\sin\theta\, c^\dagger.$$ Here $\theta$ is a parameter which belongs to the set $\left(-\frac{\pi}{4},\frac{\pi}{4}\right)\setminus\{0\}$. The reason why $\theta$ is required to be different from zero is clear: if $\theta=0$, then $a=c$ and $b=c^\dagger$, so we return to ordinary CCR. More subtle is the reason why $|\theta|$ cannot be larger or equal than $\frac{\pi}{4}$: from a mathematical point of view, in fact, when this happens the eigenstates of the number operators $N$ and $N^\dagger$ would not be square-integrable, \cite{bagswans}, Section 4.

They satisfy the pseudo-bosonic Definition \ref{def21}, with $\mathcal{D}=\mathcal{S}(\Bbb R)$ as before. The fact that these operators belong to $\mathcal{L}^\dagger(\mathcal{D})$ follows again from noticing that they are all linear combinations of $c$ and $c^\dagger$, which belong to $\mathcal{L}^\dagger(\mathcal{D})$, and from the fact that this algebra is closed for sums.

If we now define $v_j$, $j=1,2,3,4$, and $v$ as in the previous example, and we call $\mathfrak{s}$ the Lie algebra generated by these elements, and   we get the following commutator relations, related to the generators of $\mathfrak{s}$:
\begin{equation}\label{C5}
\mathfrak{s}= \langle v_1, v_2, v_3, v_4, v \ \ | \  \ [v_1,v_2]=[v_3,v_4]=v,  \ [v_1,v_4]=[v_3,v_2]=\cos2\theta\, v,
\end{equation}
$$[v_1,v_3]=[v_4,v_2]=-i\sin2\theta\, v,  \ [v,v_j]=0, \ \forall \ j=1,2,3,4 \rangle.$$
We see that the commutation rules differ from those in (\ref{C2}) for $[v_1,v_3]$ and $[v_4,v_2]$, which are no longer zero. Of course, this is true when $\sin2\theta\neq0$, which we will always assume here, to make the situation more interesting. However, it is easy to check that
\begin{equation}\label{C7}
Z(\mathfrak{s})=[\mathfrak{s},\mathfrak{s}]=\{\lambda v,\,\lambda\in\Bbb C\},
\end{equation}
so that they have both dimension one.

We proceed to show (i) and (ii). Define $\mathfrak{a} = \langle v_1, v \ | \ [v,v_1]=0 \rangle $ and clearly this is an abelian Lie algebra, then
\begin{equation}\label{AandB2}
\mathfrak{b} = \langle v_2, v_3, v_4, v \ | \ [v_4,v_2]=-i\sin2\theta\, v,  [v_3,v_2]=\cos2\theta\, v,  [v,v_2]=[v,v_3]=[v,v_4] =0 \rangle,
\end{equation}
which is non-abelian by the presence of nontrivial relations like $[v_4,v_2]=-i\sin2\theta\, v,$ and $[v_3,v_2]=\cos2\theta\, v $. Due to \eqref{C5} and \eqref{AandB2}, it is clear that $\mathfrak{a} + \mathfrak{b} = \mathfrak{s}$. On the other hand, $\mathfrak{a} \cap \mathfrak{b}= Z(\mathfrak{s})$ shows that the condition (iii) of Definition \ref{semidirect-def} is not satisfied. Therefore (i) and (ii) follow.

We proceed to show (iii). Looking at what we proved, it is enough to note
$$[\mathfrak{a}, \mathfrak{b}] \subseteq \mathfrak{a} \cap \mathfrak{b} = Z(\mathfrak{s})$$
and so both  $ \mathfrak{a}$ and  $\mathfrak{b}$ are ideals of  $\mathfrak{s}$, so (iii) follows.

Now (iv) is clear from what we have shown in (i), (ii) and (iii) above.

We proceed to show (v). As easy consequence of Definition \ref{derived}, we have already mentioned that $\mathfrak{s}/[\mathfrak{s},\mathfrak{s}]$ is abelian, so  is $\mathfrak{s}/Z(\mathfrak{s})$ by \eqref{C7}. Therefore $\mathfrak{s}$ is nilpotent of class 2. Moreover, the commutation relations show that any two elements of $\mathfrak{s}/Z(\mathfrak{s})$ commute modulo an element of $Z(\mathfrak{s})$ and these elements are exactly $v_1+Z(\mathfrak{s}), v_2+Z(\mathfrak{s}), v_3+Z(\mathfrak{s}), v_4+Z(\mathfrak{s})$. Then (v) follows.

Finally (vi) follows from the commutator relations in \eqref{C5}, which are not compatible with those of  $\mathfrak{h}(2)$ in Definition \ref{Heisenberg}  and with those of $\mathfrak{a}_{sh}$ in \eqref{C2}.

\end{proof}

The following notion shows why  (ii) of Theorem \ref{main3} is not in contrast with Definition \ref{semidirect-def}.

\begin{defn}\label{centralproducts}
A Lie algebra $\mathfrak{l}$ is the {\em central sum} of two of its Lie subalgebras $\mathfrak{a}$ and $\mathfrak{b}$ if  the following conditions are satisfied:
\begin{itemize}
\item[(i)]$\mathfrak{a}$  is an ideal of $\mathfrak{l}$,
\item[(ii)] $\mathfrak{l}= \mathfrak{a} + \mathfrak{b}$,
\item[(iii)] $\mathfrak{a} \cap \mathfrak{b} \subseteq Z(\mathfrak{a}) \cap Z(\mathfrak{b})$.
\end{itemize}
\end{defn}

Of course, the semidirect sum of two Lie algebras is a special case of the central sum of two Lie algebras, when the factors of the decomposition have only the zero element in common. There are various examples of central sums in \cite{goze} or in the long history of classifications of nilpotent Lie algebras by dimension. The reader may refer to \cite{goze, hofmor, knapp}, in order to see that we are dealing with a well known notion in geometry and algebra. Similar constructions involve in fact not only  Lie algebras, but even the structure of  compact groups (see \cite[Theorems 9.24, 9.39, 9.42]{hofmor} for details).
We conclude with a surprising fact in our context of study:

\begin{cor}The Lie algebra $\mathfrak{s}$ of Theorem \ref{main3} is a central sum of two of its Lie subalgebras.
\end{cor}

\begin{proof}
Look at the proof of Theorem \ref{main3} and take $\mathfrak{a}$ and $\mathfrak{b}$ as described there. Then one has to check that the conditions of Definition \ref{centralproducts} are satisfied.
\end{proof}

\section{conclusions and open problems}

One can involve the  notion of \textit{Schur multiplier}, introduced by I. Schur  \cite{schur} since 1907, in order to describe the behaviour of semidirect sums between Lie algebras. Even if Schur formulated this notion originally for groups, there has been a rich production in literature, especially in the context of Lie algebras, because of the generality of  methods and techniques (see \cite{nr1, nr2, nr3}). Indeed this notion involves some homological machineries which are quite technical to describe here. The reader may refer to \cite{knapp, nr1, nr2, nr3}. In particular, given a finite dimensional Lie algebra $\mathfrak{l}$ over $\mathbb{C}$, as in our case, we may define categorically the Schur multiplier of $\mathfrak{l}$ as the second Lie algebra $ H^2(\mathfrak{l},\mathbb{C}^*)$  of cohomology with complex coefficients.

There is a growing interest in geometry and physics on the notion of Schur multiplier, since it  may influence the structure of a Lie algebra, as well as that of group, that one is using in order to describe the symmetry of a dynamical system. We will quote some recent results on Schur multipliers of Lie algebras.

\begin{thm}[See \cite{nr1}, Lemmas 2.4 and 2.5] \label{ab} \

\begin{itemize}
\item[(i)]$\mathrm{dim} \ (H^2(\mathfrak{h}(1),\mathbb{C}^*))=2$.
\item[(ii)]$\mathrm{dim} \ (H^2(\mathfrak{h}(m),\mathbb{C}^*)) =2m^2-m-1$ for all $m\geq 2$.
\item[(iii)]A Lie algebra $\mathfrak{l}$ of dimension $n$ is abelian if and only if
$$\mathrm{dim}(H^2(\mathfrak{l},\mathbb{C}^*))=\frac{1}{2}n(n-1).
$$
\end{itemize}
\end{thm}

More generally, the following result is true.

\begin{thm}[See \cite{nr2}, Theorem 3.1]\label{mt}
For an $n$--dimensional non--abelian nilpotent Lie algebra $\mathfrak{l}$ such
that $\mathrm{dim} (\mathfrak{l}/[\mathfrak{l},\mathfrak{l}])=d$, we have \[\mathrm{dim}  (H^2(\mathfrak{l},\mathbb{C}^*))\leq
\frac{1}{2}(2n-d-2)\ (d-1)\ +\ 1.\]Moreover, if $d=n-1$, then the
equality holds if and only if $\mathfrak{l} \cong \mathfrak{h}(1)\oplus \mathfrak{a}$, where $\mathfrak{a}$ denotes an abelian algebra of dimension $n-3$.
\end{thm}

The following result may be proved either as an application of Theorem \ref{mt}, or via a direct computation. We will prefer this second approach.

\begin{thm}
The Schur multiplier of $\mathfrak{a}_{sh}$ has dimension $5$.
\end{thm}

\begin{proof} From Theorem \ref{main1} (v) we know that $\mathfrak{a}_{sh} \simeq \mathfrak{h}(1) \oplus \mathfrak{a}$, where $\mathfrak{a}$ is abelian of dimension 2. On the other hand,  \cite[Theorem 2.2]{nr2} implies that for a Lie algebra $\mathfrak{l}$, decomposed in the direct sum of two ideals $\mathfrak{b}$ and $\mathfrak{c}$, the following rule is true:
\[\mathrm{dim} \ (H^2( \mathfrak{b} \oplus \mathfrak{c}, \mathbb{C}^*)) = \mathrm{dim} \ (H^2(\mathfrak{b}, \mathbb{C}^*)) + \mathrm{dim}  \ (H^2(\mathfrak{c},\mathbb{C}^*)) + \mathrm{dim} \ \left(\frac{\mathfrak{b}}{[\mathfrak{b},\mathfrak{b}]} \otimes \mathfrak{c}\right).\]
This may be specialized to our case. Note that $\mathrm{dim}  \ (H^2(\mathfrak{h}(1),\mathbb{C}^*))=2$ by Theorem \ref{ab} and by the same theorem we also get  $\mathrm{dim}  \ (H^2(\mathfrak{a},\mathbb{C}^*))=1$. On the other hand,
\[\mathrm{dim}  \ \left(\frac{\mathfrak{h}(1)}{[\mathfrak{h}(1),\mathfrak{h}(1)]} \otimes \mathfrak{a}\right)= \mathrm{dim}  \ \left(\mathfrak{a} \otimes \mathfrak{a}\right) = \mathrm{dim}  \ \mathfrak{a} =2.\]
Then
\[\mathrm{dim} \ (H^2( \mathfrak{h}(1) \oplus \mathfrak{a}, \mathbb{C}^*))\]
\[ = \mathrm{dim} \ (H^2(\mathfrak{h}(1), \mathbb{C}^*)) + \mathrm{dim}  \ (H^2(\mathfrak{a},\mathbb{C}^*)) + \mathrm{dim} \ \left(\frac{\mathfrak{h}(1)}{[\mathfrak{h}(1),\mathfrak{h}(1)]} \otimes \mathfrak{a}\right)\]
\[=2 + 1 + 2 =5.\]
\end{proof}

The reason why we conclude our paper with the above results is due to the fact that we believe the pseudo-boson operators might be treated uniformly via Schur multipliers. In other words, it may happen that we have two dynamical systems, described by two different Lie algebras of pseudo-bosonic operators, but with the same structure of Schur multiplier.
For instance, we could deal with two nilpotent Lie algebras of dimension 5 as in Theorem \ref{classification}, but it might happen that they have the same Schur multiplier, so one could give an interpretation of the physical meaning of such circumstance. This might open a new line of research in the theory of pseudo-bosonic operators, involving sophistacted techniques of homology and category theory. On the other hand, the parallel problem would be to interpret the physical meaning of dynamical systems of this kind.

\section*{Acknowledgements}
F.B. acknowledges partial support from Palermo University and partial financial support from the Gruppo Nazionale di Fisica Matematica (GNFM) of the Istituto Nazionale di Alta Matematica (INdAM). F.B. also acknowledges partial financial support from the University of Cape Town, through the "Distinguished Visiting Professor" program of the Faculty of Science.


\begin{thebibliography}{99}

\bibitem{anch}J.M. Ancoch\'ea-Berm\'udez and M. Goze, Classification des alg\'ebres de Lie nilpotentes complexes de dimension 7, \textit{Arch. Math. (Basel)} \textbf{52} (1989),  175--185.


\bibitem{aitbook}  J.-P. Antoine, A. Inoue and  C. Trapani, {\it Partial $*-$algebras and Their
Operator Realizations}, Kluwer, Dordrecht, 2002.


\bibitem{baginbagbook} F. Bagarello, {\em Deformed canonical (anti-)commutation relations and non hermitian Hamiltonians}, in {Non-selfadjoint operators in quantum physics: Mathematical aspects}, F. Bagarello, J. P. Gazeau, F. H. Szafraniec and M. Znojil Eds., Wiley  (2015)




\bibitem{bagPR2015} F. Bagarello, M. Lattuca, R. Passante, L. Rizzuto and S. Spagnolo,  A Non-Hermitian Hamiltonian for a Modulated Jaynes-Cummings Model with ${\mathcal PT}$ Symmetry, {\em Phys. Rev. A} {\bf 91} (2015),  042134.


\bibitem{bagrev2007} F. Bagarello,  Algebras of unbounded operators and physical applications: a survey,  {\em Reviews in Math. Phys}  {\bf 19} (2007),  231--272.





\bibitem{bagrevA} F. Bagarello,  From self-adjoint to non self-adjoint harmonic oscillators: physical consequences
and mathematical pitfalls, \textit{Phys. Rev. A} {\bf 88} (2013),  032120.





\bibitem{bagswans} F. Bagarello,  Examples of Pseudo-bosons in quantum mechanics, \textit{ Phys. Lett. A}  {\bf 374}  (2010), 3823--3827.



\bibitem{baglat} F. Bagarello and M. Lattuca, $\mathcal{D}$ pseudo-bosons in quantum models,   \textit{Phys. Lett. A} {\bf 377} (2013), 3199--3204.





\bibitem{bag2} F. Bagarello, Appearances of pseudo-bosons from Black-Scholes equation, \textit{J. Math. Phys.} \textbf{57}
(2016), 043504.

\bibitem{bag5} F. Bagarello, A. Inoue and C Trapani, Non-self-adjoint hamiltonians defined by Riesz bases,
\textit{J. Math. Phys.} \textbf{55}  (2014), 033501.

\bibitem{bag6} F. Bagarello, Non self-adjoint Hamiltonians with complex eigenvalues, \textit{J. Phys. A}
\textbf{49}  (2016), 215304.


\bibitem{bag7} F. Bagarello, $kq$-representation for pseudo-bosons, and completeness of bi-coherent
states, \textit{J. Math. Anal. Appl.} {\bf 450} (2017), 631--643.




\bibitem{bag8} F. Bagarello, Intertwining operators for non self-adjoint Hamiltonians and bicoherent
states, \textit{J. Math. Phys.} \textbf{57} (2016), 103501.


\bibitem{dapro}N. Bebiano, J. da Provid$\mathrm{\hat{e}}$ncia and  J.P. da Provid$\mathrm{\hat{e}}$ncia,
Classes of non-hermitian operators with real eigenvalues, {\em  Electr. J. Linear Algebra} {\bf 21}  (2010), 98--109.



\bibitem{beckekolman}R.E. Beck and B. Kolman, \textit{Construction of nilpotent Lie algebras over arbitrary fields}, in: Paul S. Wang (Ed.), Proceedings of the 1981 ACM Symposium on Symbolic and Algebraic Computation, ACM, New York, 1981, pp. 169--174.



\bibitem{beltita}D. Beltit\u{a} and M. \c{S}abac, \textit{Lie Aigebras of Bounded Operators},  Birk\"auser, Berlin, 2001.

\bibitem{benjon} C.M.Bender and H.F. Jones,  Interactions of Hermitian and non-Hermitian Hamiltonians, \textit{J. Phys. A} {\bf 41}  (2008), 244006.


\bibitem{chri}O. Christensen, {\em An Introduction to Frames and Riesz Bases}, Birkh\"auser, Boston, 2003.


\bibitem{gong} M.-P. Gong, \textit{Classification of nilpotent Lie algebras of dimension 7}, PhD thesis, University of Waterloo, Waterloo, Canada, 1998.

\bibitem{goze}M. Goze and Yu. Khakimyanov, \textit{Nilpotent Lie Algebras}, Kluwer, Dordrecht, 1996.


\bibitem{degraaf}W. de Graaf, Classification of 6-dimensional nilpotent Lie algebras over field of characteristic not 2, \textit{J. Algebra}  \textbf{309} (2007), 640--653.


\bibitem{hofmor}K.H. Hofmann and S. Morris, \textit{The Structure of Compact Groups}, de Gruyter, Berlin, 2013.

\bibitem{roy2} T. K. Jana and P. Roy,  Pseudo Hermitian formulation of the quantum Black-Scholes Hamiltonian, \textit{Phys. A} {\bf 391} (2012), 2636--2640.


\bibitem{knapp}A.W. Knapp, \textit{Lie groups, Lie algebras and cohomology}, Princeton University Press, Princeton, 1988.



\bibitem{messiah} A. Messiah, {\em Quantum Mechanics}, vol. 1, North Holland Publishing Company, Amsterdam, 1967.

\bibitem{morozov}V.V. Morozov, Classification of nilpotent Lie algebras of sixth order, \textit{Izv. Vyss. Ucebn. Zaved. Mat.} \textbf{4} (1958), 161--171.


\bibitem{nr1} P. Niroomand and F.G. Russo, A restriction on the Schur multiplier of nilpotent Lie algebras,
  {\em  Electr. J. Linear Algebra} {\bf 22} (2011), 1--9.



\bibitem{nr2}
P. Niroomand and F.G. Russo, A note on the Schur multiplier of a nilpotent Lie algebra, {\em Comm.  Algebra},
\textbf{39} (2011), 1293--1297.

\bibitem{nr3} P. Niroomand, M. Parvizi and F.G. Russo, Some criteria for detecting capable Lie algebras, \textit{J. Algebra} \textbf{384} (2013), 36--64.

\bibitem{nr4}P. Niroomand and F.G. Russo, Some restrictions on the Betti numbers of a nilpotent Lie algebra, \textit{Bull. Belg. Math. Soc. - Simon Stevin} \textbf{21} (2014), 403--413.

\bibitem{schu} K. Schm\"udgen, {\it Unbounded operator algebras and Representation theory}, Birkh\"auser, Basel, 1990.


\bibitem{schur} I. Schur, Untersuchungen \"uber die Darstellung der endlichen Gruppen durch gebrochene lineare Substitutionen,  \textit{J. Reine Angew. Math} \textbf{132} (1907), 85--137.



\bibitem{snobl}L. Snobl and P. Winternitz, \textit{Classification and identification of Lie algebras},
CRM Monograph Series, 33. American Mathematical Society, Providence, RI, 2014.


\bibitem{swans} M.S. Swanson,  Transition elements for a non-Hermitian quadratic Hamiltonian, {\em J. Math. Phys.} {\bf 45} (2004), 585--601.


\bibitem{gap} The GAP Group, GAP—Groups, algorithms, and programming, version 4.4, http://www.gap-system.org, 2004.


\bibitem{trrev} C. Trapani, Quasi $*-$algebras of operators and their applications,  \textit{ Reviews Math. Phys.} \textbf{ 7} (1995),  1303--1332.



\bibitem{turk}P. Turkowski, Solvable Lie algebras of dimension six, \textit{ J. Math. Physics} \textbf{31} (1990), 1344--1350.
















\end{thebibliography}
\end{document}